\documentclass[11pt]{article}
\usepackage{booktabs} 
\usepackage[ruled]{algorithm2e} 
\usepackage{tablefootnote}
\usepackage{bbm}
\usepackage{tikz}
\usetikzlibrary{fit,arrows,calc,positioning}

\SetAlFnt{\small}
\SetAlCapFnt{\small}
\SetAlCapNameFnt{\small}
\SetAlCapHSkip{0pt}
\IncMargin{-\parindent}

\usepackage[utf8]{inputenc} 
\usepackage[T1]{fontenc}    
\usepackage{hyperref}       
\usepackage{url}            
\usepackage{booktabs}       
\usepackage{amsfonts}       
\usepackage{nicefrac}       
\usepackage{microtype}      
\usepackage{lipsum}
\usepackage{graphicx}
\graphicspath{{media/}}     
\usepackage[colorinlistoftodos]{todonotes}

\usepackage{amsmath,amsthm}
\usepackage{cleveref}
\usepackage{bbm}
  
\usepackage{graphicx}
\usepackage{subfig}
 
\newtheorem{theorem}{Theorem}[section]

\newtheorem{corollary}[theorem]{Corollary}

\newtheorem{lemma}[theorem]{Lemma}
\newtheorem{claim}[theorem]{Claim}

\newtheorem{assume}[theorem]{Assumption}
\newtheorem{remark}[theorem]{Remark}
\newtheorem{definition}[theorem]{Definition}

\newtheorem{proposition}[theorem]{Proposition}

\usepackage[utf8]{inputenc}
\usepackage{fullpage}

\usepackage{graphicx}
\usepackage{subfig}
\usepackage[colorinlistoftodos]{todonotes}
\renewcommand{\vec}{\mathbf }
\renewcommand{\cite}{\citet}
\usepackage{natbib}
\usepackage{hyperref}

\def\@fnsymbol#1{\ensuremath{\ifcase#1\or *\or **\or
   \mathsection\or \mathparagraph\or \|\or **\or \else\@ctrerr\fi}}
   
\definecolor{cadmiumgreen}{rgb}{0.2, 0.62, 0.54}

\title{Incentive Compatibility in the Auto-bidding World}

   \makeatletter
\def\@fnsymbol#1{\ensuremath{\ifcase#1\or \dagger\or *\or
   \mathsection\or \mathparagraph\or \|\or **\or \dagger\dagger
   \or \ddagger\ddagger \else\@ctrerr\fi}}
    \makeatother

\author{Yeganeh Alimohammadi\thanks{Stanford University, \texttt{yeganeh@stanford.edu}}, Aranyak Mehta\thanks{Google, \texttt{\{aranyak,perlroth\}@google.com}}\; and  Andres Perlroth\footnotemark[2]}

\begin{document}

\begin{titlepage}

\maketitle
\begin{abstract}

Auto-bidding has recently become a popular feature in ad auctions. This feature enables advertisers to simply provide high-level constraints and goals to an automated agent, which optimizes their auction bids on their behalf. These auto-bidding intermediaries interact in a decentralized manner in the underlying auctions, leading to new interesting practical and theoretical questions on auction design, for example, in understanding the bidding equilibrium properties between auto-bidder intermediaries for different auctions. In this paper, we examine the effect of different auctions on the incentives of advertisers to report their constraints to the auto-bidder intermediaries. More precisely, we study whether canonical auctions such as first price auction (FPA) and second price auction (SPA) are {\em auto-bidding incentive compatible (AIC)}: whether an advertiser can gain by misreporting their constraints to the autobidder.

We consider value-maximizing advertisers in two important settings: when they have a budget constraint and when they have a target cost-per-acquisition constraint. The main result of our work is that for both settings, FPA and SPA are not AIC. This contrasts with FPA being AIC when auto-bidders are constrained to bid using a (sub-optimal) uniform bidding policy. 
We further extend our main result and show that any (possibly randomized) auction that is truthful (in the classic profit-maximizing sense), scale-invariant and symmetric is not AIC. Finally, to complement our findings, we provide sufficient market conditions for FPA and SPA to become AIC for two advertisers. These conditions require advertisers' valuations to be well-aligned. This suggests that when the competition is intense for all queries, advertisers have less incentive to misreport their constraints. 

From a methodological standpoint, we develop a novel continuous model of queries. This model provides tractability to study equilibrium with auto-bidders, which contrasts with the standard discrete query model, which is known to be hard. Through the analysis of this model, we uncover a surprising result: in auto-bidding with two advertisers, FPA and SPA are auction equivalent.
\end{abstract}
\end{titlepage}


\section{Introduction}

Auto-bidding has become a popular tool in modern online ad auctions, allowing advertisers to set up automated bidding strategies to optimize their goals subject to a set of constraints. By using algorithms to adjust the bid for each query, auto-bidding offers a more efficient and effective alternative to the traditional fine-grained bidding approach, which requires manual monitoring and adjustment of the bids.

There are three main components in the auto-bidding paradigm: 1) the advertisers who provide high-level constraints to the auto-bidders,  2) the auto-bidder agents who bid -- in a decentralized manner -- on behalf of each advertiser to maximize the advertiser's value subject to their constraints, and 3) the query-level auctions where queries are sold (see Figure~\ref{fig:autobidding}).

Current research has made significant progress in studying the interactions of the second and third components in the auto-bidding paradigm, particularly in understanding equilibrium properties (e.g., welfare and revenue) between the auto-bidders intermediaries for different auction rules \citep{Aggarwal19, BalseiroRobustAuctionDesign, deng2021towards, mehta2021, liaw2022}. There is also work on mechanism design for this setting in more generality, i.e., between the advertisers and the auctioneer directly abstracting out the second component \citep{balseiro2021landscape, BalseiroOptimalClipping,GolrezaiLP21}.

Our work, instead, examines the relation between {value-maximizing}
advertisers, who maximize the value they obtain subject to a payment constraint, 
and the other two components of the auto-bidding paradigm. 
More precisely, we study the impact of different auction rules on the incentives of advertisers to report their constraints to the auto-bidder intermediaries. We specifically ask whether canonical auctions such as first price auction (FPA), second price auction (SPA), and general truthful auctions are {\em auto-bidding incentive compatible} (AIC) - in other words, can advertisers gain by misreporting their constraints to the auto-bidder?

We consider value-maximizing advertisers in two important settings: when they have a budget constraint and when they have a target cost-per-acquisition (tCPA) constraint\footnote{The former is an upper bound on the total spend, and the latter is an upper bound on the average spend per acquisition (sale). Our results clearly hold for more general auto-bidding features, such as target return on ad-spend (tRoAS) where the constraint is an upper bound on the average spend per value generated. }. The main result of our work is that for both settings, FPA and SPA are not AIC. This contrasts with FPA being AIC when auto-bidders are constrained to bid using a (sub-optimal) uniform bidding policy.
 We further generalize this surprising result and show that any (possibly randomized) truthful auction having a scale invariance and symmetry property is also not AIC.  We complement our result by providing sufficient market conditions for FPA and SPA to become AIC for two advertisers.
These conditions require advertisers' valuations to be well-aligned. This suggests that advertisers have less incentive to misreport their constraints when the competition is intense for all queries.

In our model, each advertiser strategically reports a constraint (either a tCPA or a budget) to an auto-bidder agent, which bids optimally on their behalf in each of the queries. Key in our model, we consider a two-stage game where first, advertisers submit constraints to the auto-bidders and, in the subgame, auto-bidders reach a bidding equilibrium across all query auctions. 
Thus, the whole bidding subgame equilibrium can change when an advertiser deviates and reports a different constraint to its auto-bidder.\footnote{This two-stage model captures the idea that auto-bidding systems rapidly react to any change in the auction. Hence, if there is any change in the bidding landscape, auto-bidders quickly converge to a new equilibrium.} In this context, an auction rule is called auto-bidding incentive compatible (AIC) if, for all equilibria, it is optimal for the advertiser to report their constraint to the auto-bidder.

\tikzstyle{b} = [rectangle, draw, fill=blue!20, node distance=.2cm and 2.2cm, text width=3.8em, text centered, rounded corners, minimum height=4em, thick]
\tikzstyle{bb} = [rectangle, draw, fill=blue!20, node distance=.8cm, text width=3.5em, text centered, rounded corners, minimum height=13em, thick]
\tikzstyle{c} = [rectangle, draw, inner sep=0.5cm, dashed]
\tikzstyle{l} = [draw, -latex',thick]
\begin{figure}
\centering
\begin{tikzpicture}[auto]
    \node [b] (ad2) {\footnotesize Advertiser};
       \node [b,below=of ad2](ad3){\footnotesize Advertiser};
       \node [b,above=of ad2](ad1){\footnotesize Advertiser};
       \node [b,right=of ad2](ab2){\footnotesize Auto-bidder Agent};
       \node [b,right=of ad3](ab3){\footnotesize Auto-bidder Agent};
       \node [b,right=of ad1](ab1){\footnotesize Auto-bidder Agent};
       \node [bb, right=of ab2](auction){\footnotesize Auction per query};

    \path [l] (ad1) -- (ab1);
    \path [l] (ad2) -- (ab2);
    \path [l] (ad3) -- (ab3);
        \draw (ad1) -- (ab1) node  [midway, above, sloped,
        align=left, ] (TextNode) {\footnotesize  constraints};
            \draw (ad3) -- (ab3) node  [midway, above, sloped,
            align=left, ](TextNode) {\footnotesize constraints};
    \draw (ad2) -- (ab2) node [midway, above, sloped,
    align=left, ](t) {\footnotesize constraints};
\path [l] (ab1) -- (ab1-|auction.west);
    \path [l] (ab2) -- (ab2-|auction.west);
    \path [l] (ab3) -- (ab3-|auction.west);
        \draw (ab1) -- (ab1-|auction.west) node [midway, above, sloped] (TextNode) {\footnotesize bids};    
        \draw (ab2) -- (ab2-|auction.west) node [midway, above, sloped] (TextNode) {\footnotesize bids};    
        \draw (ab3) -- (ab3-|auction.west) node [midway, above, sloped] (TextNode) {\footnotesize bids};
\end{tikzpicture}
\caption{The Auto-bidding Process: Advertisers submit constraints and receive query allocations with specified costs as output. Inside the auto-bidding feature, each advertiser has an agent that optimizes bidding profile within each advertiser's constraints.} \label{fig:autobidding}
\end{figure}
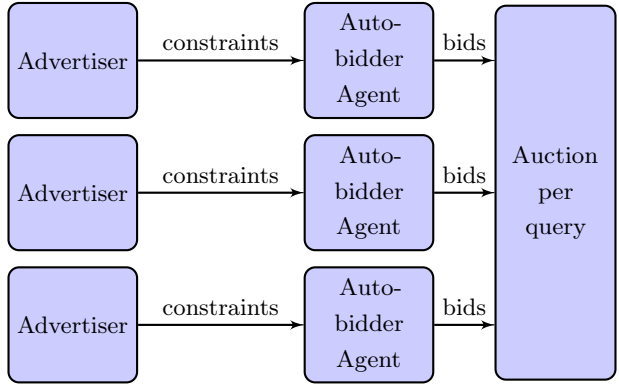

\subsection{Main Results}

 We begin our results by presenting a stylized example in Section~\ref{sec: warmup} that demonstrates how auto-bidding with SPA is not AIC (Theorem~\ref{thm: second price Warm up}). Our example consists of a simple instance with three queries and two advertisers. This example highlights a scenario where an advertiser can benefit from lowering their reported budget or tCPA-constraint.

We then introduce a continuous query model that departs from the standard auto-bidding model by considering each query to be of infinitesimal size. This model provides tractability in solving equilibrium for general auction rules like FPA, which is key to studying such auctions' auto-bidding incentive compatibility properties.  
Further, this continuous-query model succinctly captures 
real-world scenarios where the value of a single query is negligible compared to the pool of all queries that are sold.

Under the continuous-query model, we study the case where queries are sold using FPA and show that in the auto-bidding paradigm, FPA is not AIC (Section~\ref{sec: first price genral}). We first characterize the optimal bidding strategy for each auto-bidder agent, which, surprisingly, has a tractable form.\footnote{ Notice that in the discrete-query model, there is no simple characterization for the auto-bidder best response in an FPA.} We then leverage this tractable form to pin down an equilibrium for the case of two auto-bidders when both auto-bidders face a budget or tCPA constraint. In this equilibrium, the seller divides queries between the two advertisers based on the ratio of their values for each advertiser. Specifically, advertiser 1 receives queries for which the ratio of its value to the other advertiser's value is higher than a certain threshold. From this point, determining the equilibrium reduces to finding a threshold that tightens advertisers' constraints (see Lemma~\ref{lm: equilibrium first price} for more detail).
We then show that for instances where the threshold lacks monotonicity with the auto-bidders constraints, advertisers have the incentive to misreport the constraint to the auto-bidder (Theorem~\ref{thm: first price general}). Conversely, when the thresholds are monotone, advertisers report constraints truthfully. We show conditions on the advertisers' valuations for the two-advertisers setting to guarantee this monotonicity (Theorem~\ref{thm: sufficient condition-first price}). This condition requires a strong positive correlation of the advertisers' valuations across the queries. As a practical insight, our results suggest that for settings where the competition on all queries is intense, advertisers' incentives to misreport are weak.

We then explore the case where, in FPA, auto-bidders are constrained to bid using a {\em uniform bidding strategy:} the bid on each query is a constant times the advertiser's value.\footnote{ Uniform bidding strategy is also known in the literature as pacing bidding \citet{conitzer2022pacing, chen2021complexity,conitzer2022multiplicative,gaitonde2022budget}. 
} Uniform bidding is only optimal when auctions are truthful \citep{Aggarwal19}. Even though for FPA, these strategies are suboptimal, {they have gained recent attention in the literature due to their tractability \citet{conitzer2022pacing, conitzer2022multiplicative,chen2021complexity,gaitonde2022budget}.}
 We show that in such a scenario, FPA with uniform bidding turns out to be AIC (Theorem~\ref{thm:first price pacing}). {However, while this proves AIC in our model, the suboptimality of uniform bidding for FPA can give rise to incentives to deviate in other ways outside our model, e.g., by splitting the advertising campaigns into multiple campaigns with different constraints. These considerations are important when implementing this rule in practice.} 

The second part of the paper pivots to the case where auctions are truthful, i.e., auctions in which it is optimal for a profit-maximizing agent to bid their value. We first study the canonical SPA and show that, in our continuous-query model, SPA and FPA are auction equivalent. The allocation and payments among the set of {\em reasonable} equilibria hold equal (Theorem~\ref{thm: equivalent first and second}).\footnote{We show the auction equivalence among uniform bidding equilibria for SPA and threshold-type equilibrium for FPA.} As a Corollary, the results we obtain for FPA apply to SPA as well: SPA is not AIC, and we derive sufficient conditions on advertisers' valuations so that SPA is AIC for two advertisers. We then consider a general class of randomized truthful auctions. We show that if the allocation rule satisfies these natural conditions: \footnote{ These conditions have been widely studied in the literature due to their practical use \citep{mehta2021, liaw2022,allouah2020prior}.} (i) scaled invariant (if all bids are multiplied by the same factor then the allocation doesn't change), and (ii) is symmetric (bidders are treated equally); then the auction rule is not AIC. The main results are summarized in Table~\ref{tab:result}.
\begin{table}[!h]
    \centering
    \begin{tabular}{|c|c|}\hline
      Per Query Auction   & AIC \\\hline\hline
         Second-Price Auction& Not AIC\\   Truthful Auctions & Not AIC\\   First-Price Auction & Not AIC\\ First-Price Auction with Uniform Bidding & AIC\tablefootnote{As previously discussed, implementing FPA with the suboptimal uniform bidding policy can create other distortion on advertisers' incentives (e.g., splitting their campaign into multiple campaigns with different constraints). }\\ \hline
    \end{tabular}
    \caption{Main Results}
    \label{tab:result}
\end{table}

\subsection{Related Work}\label{sec: litreview}

The study of auto-bidding in ad auctions has gained significant attention in recent years. One of the first papers to study this topic is \citet{Aggarwal19}, which presents a mathematical formulation for the auto-bidders problem given fixed constraints reported by advertisers. They show that uniform bidding is an optimal strategy if and only if auctions are truthful (in the profit-maximizing sense). They further started an important line of work to measure, using a Price of Anarchy (PoA) approach, the welfare implications when auto-bidders are bidding in equilibrium for different auctions. Current results state that for SPA, the PoA is $2$ \citet{Aggarwal19} and also for FPA \citet{liaw2022}\footnote{The authors show that for a general class of deterministic auctions $PoA\geq 2$.}, and, interestingly, it can be improved if the auction uses a randomized allocation rule \citet{mehta2021, liaw2022}. In a similar venue, \citet{deng2021towards, balseiro2021robust} studies models where the auction has access to extra information and show how reserves and boosts can be used to improve welfare and efficiency guarantees.

A second line of work studies how to design revenue-maximizing auctions when bidders are value-maximizing agents and may have private information about their value or their constraints \citep{GolrezaiLP21, balseiro2021landscape,balseiro2021robust}.
{In all these settings, the 
mechanism designer is not constrained to the presence of the {auto-bidding intermediaries} (Component 2 in Figure~\ref{fig:autobidding}). Our study has added structure by having advertisers submit their constraints first, followed by a decentralized subgame to achieve a bidding equilibrium before allocating the queries and determining payments.}
Thus, prior, their mechanism setting can achieve broader outcomes than our auto-bidding constraint paradigm. Interestingly, for the one query case, the authors show that FPA with a uniform bidding policy is optimal \citet{balseiro2021landscape}. Our results complement theirs and show that such a mechanism is implementable in the auto-bidding constraint paradigm and is AIC.  

Closer to our auto-bidding paradigm, a recent line of work has started to study the incentive of advertisers when bidding via an auto-bidder intermediary. 
\citet{kolumbus2022auctions,kolumbus2022how} study regret-minimizing advertisers who can strategically report their private values of queries to the auto-bidder. Similar to our findings, they show that SPA is not incentive compatible in the sense that advertisers have the incentive to misreport their true value to get a lower regret over repeated auctions. In contrast, they prove FPA with two advertisers and a single query is incentive compatible. Our work complements their findings by considering a more general setting that allows for multiple queries and involves advertisers reporting high-level constraints such as budget and tCPA rather than the value of the query itself. In a different study that emerged after the initial version of this paper, \cite{feng2023strategic} delve into the dynamics of the game played among advertisers in an auto-bidding setting with FPA. They bound the price of anarchy of the game between advertisers, where they strategically declare their budget constraints, to a factor of 4.
In a different study,
\citet{mehtaperlroth22} show that a profit-maximizing agent may benefit by reporting a target-based bidding strategy to the auto-bidder when the agent has concerns that the auctioneer may change (ex-post) the auction rules. Also, in empirical work, \citet{li2022auto} develop a new methodology to numerically approximate 
auto-bidding equilibrium and show numerical examples where advertisers may benefit by reporting their constraints on SPA. Our work complements their findings by showing that SPA is not AIC under a theoretical framework.

Our work also connects with the literature about auctions with budgeted constraint bidders. 
In particular, our results are closely related to \citet{conitzer2022pacing} who study FPA with uniform bidding (a.k.a. pacing bidding).  They introduce the concept of the first-price auction pacing equilibrium (FPPE) for budget-constrained advertisers, which is the same as the equilibrium in our auto-bidding subgame. 
They show that in FPPE the revenue and welfare are monotone increasing as a function of the advertisers' budgets.  In our work, we show that in FPPE,  advertisers' \textit{values} are monotone as a function of their reported budget.
In addition, they differentiate between first and second-price by showing that FPPE is computable, unlike SPPE, where maximizing revenue has previously been known to be NP-hard \citet{conitzer2022multiplicative}, and that the general problem of approximating the SPPE is PPAD-complete \citet{chen2021complexity}. In contrast, we show that both SPA and FPA are tractable in the continuous model. 
Interestingly, this dichotomy between FPA and SPA (both with uniform bidding) is also reflected in our work -- the former is AIC, while the latter is not.

Uniform bidding has been explored in a separate body of research on repeated auctions without the presence of auto-bidding. \citet{balseiro2019learning} investigate strategies to minimize regret in simultaneous first-price auctions with learning. \citet{gaitonde2022budget} take this concept further by extending the approach to a wider range of auction settings. Furthermore, \citet{golrezaei2021bidding} examines how to effectively price and bid for advertising campaigns when advertisers have both ROI and budget constraints.

\section{Warm Up: Second Price Auction is not AIC!}\label{sec: warmup}

To understand the implications of the auto-bidding model, we start with an example of auto-bidding with the second-price auction. 
Through this example, we will demonstrate the process of determining the equilibrium in an auto-bidding scenario and emphasize a case where the advertiser prefers to misreport their budget leading to the following theorem.

\begin{theorem}\label{thm: second price Warm up}For the budget setting (when all advertisers are budgeted-constrained) and for the tCPA-setting (when all advertisers are tCPA-constrained), we have that SPA is not AIC. That is, there are some instances where an advertiser benefits by misreporting its constraint.
\end{theorem}

\begin{proof}{Proof.}
Consider two budget-constrained advertisers and four queries $Q=\{q_1,q_2,q_3,q_4\}$, where the expected value of winning query $q$ for advertiser $a$ is denoted by $v_a(q)$, and it is publicly known (as in Table \ref{tab:budget-second price}).  The actual budgets of advertisers are $B_1=20$ and $B_2=40$.
At first, assume each advertiser reports their budget to the auto-bidder. Then the auto-bidder agents, one for each advertiser, submit the bidding profiles (to maximize their advertisers' value subject to the budget constraint).  
The next step is a second-price auction per query, where the queries are allocated to the highest bidder.

 \begin{table}[!htp]
    \centering
    \begin{tabular}{c|c|c|c|c}
         & $q_1$& $q_2$& $q_3$& $q_4$ \\\hline
        Advertiser 1 &  2.1 & 40 & 30 &20\\
        Advertiser 2 & 1 & 20 & 25 &100\\
    \end{tabular}
    \caption{SPA with two budget constraint advertisers is not AIC: The value of each query for each advertiser.}
    \label{tab:budget-second price}
\end{table}

Finding the equilibrium bidding strategies for the auto-bidder agents is challenging, as the auto-bidder agents have to find the best-response bids vs. the response of other auto-bidder agents, and each auto-bidder agent's bidding profile changes the cost of queries for the rest of the agents. We use the result of \cite{Aggarwal19}, which establishes that uniform bidding is an almost-optimal strategy in any truthful auto-bidding auction, including SPA. In Appendix~\ref{appendix: tcpa second price warm up},  we show that their result implies that uniform bidding is the optimal strategy when the query values are given as in Table~\ref{tab:budget-second price}. So, without loss of generality, we can assume all auto-bidding agents adopt uniform bidding. Consequently, each auto-bidding agent optimizes over one variable, a bidding multiplier $\mu_a$, and then bids on query $q$ with the scaled value $\mu_av_a(q)$.

We show that with the given budgets $B_1=20$ and $B_2=49$, an equilibrium exists such that advertiser $1$  wins  $q_1$ and $q_2$. To see this, let $\mu_1=.7$ and $\mu_2=.91$. We need to check: (i) Allocation: advertiser 1 bids higher on $q_1$ and $q_2$ and less on $q_3$ and $q_4$ compared to advertiser $2$, (ii) 
Budget constraints are satisfied, and (iii) Bidding profiles are the best response: The auto-bidder agents do not have the incentive to increase their multiplier to get more queries. These three conditions can be checked as follows:
 \begin{enumerate}
    \item \emph{Allocation inequalities:} For each query, the advertiser with the highest bid wins.
\[ \frac{v_1(q_1)}{v_2(q_1)}>\frac{v_1(q_2)}{v_2(q_2)}> \frac{\mu_2}{\mu_1}> \frac{v_1(q_3)}{v_2(q_3)}> \frac{v_1(q_4)}{v_2(q_4)}.\]
        \item
       \emph{ Budget constraints:} Since the auction is second-price the cost of query $q$ for advertiser 1 is $\mu_2v_2(q)$ and for advertiser 2 is $\mu_1v_1(q)$. So, we must have the following inequalities to hold so that the budget constraints are satisfied:
    \[20=B_1\geq \mu_2 (v_2(q_1)+v_2(q_2))=19.11\qquad \text{(Advertiser 1)},\]
    \[49=B_2\geq \mu_1\left(v_1(q_3)+v_1(q_4)\right)=35\qquad \text{(Advertiser 2)}.\]
\item \emph{Best response:} Does the advertiser's agent have the incentive to raise their multiplier to get more queries? If not, they shouldn't afford the next cheapest query.
    \[20=B_1<  \mu_2\left(v_2(q_1)+v_2(q_2)+v_2(q_3)\right)=41.86\qquad \text{(Advertiser 1)},\]
    \[49=B_2< \mu_1\left(v_1(q_2)+v_1(q_3)+v_1(q_4)\right)=63\qquad \text{(Advertiser 2)}.\]
    \end{enumerate}
    Since all three conditions are satisfied, thus, this profile is an equilibrium for the auto-bidders bidding game, where advertiser 1 wins  $q_1$ and $q_4$, and advertiser 2 wins $q_2$ and $q_3$.

Now, consider the scenario in which advertiser 1 wants to strategically report their budget to the auto-bidder.  Intuitively, the budget constraint for the auto-bidder agent should be harder to satisfy, and hence the advertiser should not win more queries. But, contrary to this intuition, when advertiser $1$ reports a lower budget $B_1'=10$, we show that there exists a unique auto-bidding equilibrium such that advertiser 1 wins $q_1$, $q_2$ and $q_4$ (more queries than the worst case equilibrium where advertiser 1 reports their actual budget $B_1=20$). Here, the uniqueness is with respect to the allocation of queries. 
Similar to above,  we can check that $\mu_1'=1$, and $\mu_2'=\frac{B_1}{v_2(q_1)+v_2(q_2)+v_2(q_3)}=\frac{10}{46}$ results in an equilibrium (we prove the uniqueness in Appendix~\ref{appendix: tcpa second price warm up}):
\begin{enumerate}
    \item Allocation: advertiser 1 wins $q_1$, $q_2$ and $q_4$ since it has a higher bid on them,
\[ \frac{v_1(q_1)}{v_2(q_1)}> \frac{v_1(q_2)}{v_2(q_2)}>\frac{v_1(q_3)}{v_2(q_3)}> \frac{\mu'_2}{\mu'_1}>  \frac{v_1(q_4)}{v_2(q_4)}.\]
    \item Budget constraints:
    \[10=B_1\geq  \mu_2'\left(v_2(q_1) +v_2(q_2)+v_2(q_3)\right)=10, \qquad\text{and }\]
    \[49=B_2\geq \mu_1'v_1(q_4)=20\]
\item Best response: 
    \[10=B_1< \mu_2'\left(v_2(q_1)+v_2(q_2)+v_2(q_3)+v_2(q_4)\right)\approx 31.7,\qquad\text{and}\]
    \[49=B_2< \mu_1'\left(v_1(q_3)+v_1(q_4)\right)=50\]
\end{enumerate} In Appendix~\ref{appendix: tcpa second price warm up}, we will show this equilibrium is unique and also give a similar non-AIC example for the case of tCPA-constrained advertisers.
\end{proof}

Before studying other canonical auctions, in the next section, we develop a tractable model of continuous query. Under this model,, the characterization of the auto-bidders bidding equilibria when the auction is not SPA is tractable. This tractability is critical for studying auto-bidding incentive compatibility.

\section{Model}\label{sec: model}

The baseline model consists of a set of $A$ advertisers competing for $q\in Q$ single-slot queries owned by an auctioneer. We consider a continuous-query model where $Q=[0,1]$. Let $x_a(q)$ be the probability of winning query $q$ for advertiser $a$. Then the expected value and payment of winning query $q$  at price $p_a(q)$ are $x_a(q) v_a(q)dq$ and $p_a(q)dq$.\footnote{All functions $v_a, x_a, p_a$ are integrable with respect to the Lebesgue measure $dq$. 
}$^,$ \footnote{The set $Q=[0,1]$ is chosen to simplify the exposition. Our results apply to a general metric measurable space $(Q,\mathcal{A}, \lambda)$ with atomless measure $\lambda$.} Intuitively, this continuous-query model is a continuous approximation for instances where the size of each query relative to the whole set is small. 

The auctioneer sells each query $q$ using a query-level auction which induces the allocation and payments $(x_a(q), p_a(q))_{a\in A}$ as a function of the bids $(b_a)_{a \in A}$. In this paper, we focus on the  First Price Auction (FPA), Second Price Auction (SPA), and, more generally, any Truthful Auction (see Section~\ref{sec: truthful} for details).

 \subsection*{Auto-bidder agent:}
 Advertisers do not participate directly in the auctions; rather, they report high-level goal constraints to an auto-bidder agent who bids on their behalf in each query. Thus, Advertiser $a$ reports a budget constraint $B_a$ or a target cost-per-acquisition constraint (tCPA) $T_a$ to the auto-bidder. Then, the auto-bidder taking as fixed other advertiser's bid submits bids $b_a(q)$ to induce $x_a(q),p_a(q)$ that solves
\begin{align}\label{eq: auto-bidder agent-descrete}
    \max& \int_0^1 x_a(q) v_a(q) dq\\
    \text{s.t.}& \int_0^1 x_a(q)p_a(q)dq\leq B_a+T_a \int_0^1 x_a(q) v_a(q) dq.\label{eq: constraint}
\end{align}

The optimal bidding policy does not have a simple characterization for a general auction. However, when the auction is truthful (like SPA) the optimal bid takes a simple form in the continuous model. \citep{Aggarwal19}.
\begin{remark}[Uniform Bidding]
 If the per-query auction is truthful, then uniform bidding is the optimal policy for the autobidder. Thus, $b_a(q) = \mu \cdot v_a(q)$ for some $\mu >0$. We formally prove this in Claim~\ref{claim: uniform bididng}.
\end{remark}
\begin{remark}[Information on Value Functions]
  In our model, auto-bidders possess knowledge of their agents' valuation functions for the queries, which mainly implies that each auto-bidder has a good sense of the bid landscape they face. In practice, auto-bidders have interacted in multiple auction rounds where they experiment and learn. Thus, after enough rounds of experimentation, each auto-bidder has a good sense of their bid landscape. Therefore, our model captures how auto-bidders behave once they have learned enough about the market through experimentation.  This assumption is fairly standard in the auto-bidding literature (see e.g., \cite{Aggarwal19,feng2023strategic}).
\end{remark}

\subsection*{Advertisers}

Following the current paradigm in auto-bidding, we consider that advertisers are value-maximizers and can be of two types:  a budget-advertiser or tCPA-advertiser. The payoffs for these advertisers are as follows.
\begin{itemize}
\item For a budget-advertiser with budget $B_a$, the payoff is
$$ u_a = \begin{cases}
\int_{0}^1 x_a(q) v_a(q)dq &\mbox{ if } \int_{0}^1 p_a(q)dq\leq B_a \\
-\infty & \mbox{ if not}.
\end{cases}
$$
\item For a tCPA-advertiser with target $T_a$, the payoff is
$$ u_a = \begin{cases}
\int_{0}^1 x_a(q) v_a(q)dq &\mbox{ if } \int_{0}^1 p_a(q)dq\leq T_a \cdot \int_{0}^1 x_a(q) v_a(q)dq \\
-\infty & \mbox{ if not}.
\end{cases}
$$
\end{itemize}

\subsection*{Game, Equilibrium and Auto-bidding Incentive Compatibility (AIC)}

The timing of the game is as follows. First, each advertiser depending on their type, submits a budget or target constraint to an auto-bidder agent. Then, each auto-bidder solves Problem~\ref{eq: auto-bidder agent-descrete} for the respective advertiser. Finally, the per-query auctions run, and allocations and payments accrue. 

We consider a complete information setting and use subgame perfect equilibrium (SPE) as a solution concept. Let $V_a(B'_a; B_a)$ the expected payoff in the subgame for a budget-advertiser with budget $B_a$ that reports $B'_a$ to the auto-bidder (likewise, we define $V_a(T'_a;T_a)$ for the tCPA-advertiser).

\begin{definition}[Auto-bidding Incentive Compatibility (AIC)]
An auction rule is Auto-bidding Incentive Compatible (AIC) if for every SPE we have that $V_a(B_a;B_a) \geq V_a(B'_a; B_a)$ and $V_a(T_a;T_a) \geq V_a(T'_a; T_a)$ for every $B_a,B'_a,T_a, T'_a$.
\end{definition}

Similar to the classic notion of incentive compatibility, an auction rule satisfying AIC makes the advertisers' decision simpler: they need to report their target to the auto-bidder. However, notice that the auto-bidder plays a subgame after the advertiser's reports. Thus, when Advertiser $a$ deviates and submits a different constraint, the subgame outcome may starkly change not only on the bids of Advertiser $a$ but also on other advertisers' bids may change.

\section{First Price Auctions}\label{sec: first price genral}
This section demonstrates that the first price auction is not auto-bidding incentive compatible. 

\begin{theorem}\label{thm: first price general} 
 Suppose there are at least two budget-advertisers or two tCPA-advertisers; then FPA is not AIC. 
 \end{theorem}

Later in Section~\ref{sec: sufficient conditions}, we show a complementary result by providing sufficient conditions on advertisers' value functions to make FPA AIC for the case of two advertisers. We show that this sufficient condition holds in many natural settings, suggesting that in practice, FPA tends to be AIC.    

Then in Section~\ref{sec: pacing}, we turn our attention to FPA, where auto-bidders are restricted to use uniform bidding across the queries. In this case, we extend to our continuous-query model the
result of \cite{conitzer2022pacing} and show the following result.

\begin{theorem}\label{thm:first price pacing}
FPA restricted to uniform bidding is AIC.
\end{theorem}

\subsection{Proof of Theorem~\ref{thm: first price general}}\label{sec: FPA is not AIC}

We outline the proof of Theorem~\ref{thm: first price general} by dividing it into three key steps. Step 1 characterizes the best response bidding profile for an auto-bidder.
Determining the best response in FPA is known to be hard \cite{filos2021complexity}. However, we establish a connection between first-price and second-price auctions in the continuous query model. Specifically, we demonstrate that in FPA with continuous queries, each advertiser wins queries with payments as they would win if the auction were a second price auction (SPA) with a uniform bidding strategy.
This observation simplifies the problem from optimizing bids for each query to determining a single multiplicative variable for each advertiser.

In Step 2, we leverage the tractability of our continuous-query model and pin down the sub-game bidding equilibrium when there are either two budget-advertisers or two tCPA-advertisers in the game (Lemma~\ref{lm: equilibrium first price}). We derive an equation that characterizes the ratio of the multipliers of each advertiser as a function of the constraints submitted by the advertisers.
This ratio defines the set of queries that each advertiser wins, and as we will see, the value accrued by each advertiser is monotone in this ratio. So, to find a non-AIC example, one has to find scenarios where the equilibrium ratio is not a monotone function of the input constraints, which leads to the next step. 

To conclude, we show in Step 3 an instance where the implicit solution for the ratio is nonmonotonic, demonstrating that auto-bidding in first-price auctions is not AIC.
As part of our proof, we interestingly show that AIC is harder to satisfy when advertisers face budget constraints rather than tCPA constraints
(see Corollary~\ref{cor: tCPA implies budget}).

\subsection*{Step 1: Optimal Best Response}

The following claim shows that, contrary to the discrete-query model, the best response for the auto-bidder in a first-price auction can be characterized as a function of a single multiplier. 

\begin{claim}\label{claim: first price best response}
Taking other auto-bidders as fixed, there exists 
 a multiplier $\mu_a\geq 0$ such that the following bidding strategy is optimal:
\[b_a(q)=\begin{cases}\max_{a'\neq a}\left(b_{a'}(q)\right)& \mu_av_a(q)\geq \max_{a'\neq a}\left(b_{a'}(q)\right)\\
0& \mu_av_a(q)\neq \max_{a'}\left(b_{a'}(q)\right).
\end{cases}\]
The result holds whether the advertiser is budget-constrained or tCPA-constrained\footnote{
In FPA, ties are broken consistently with the equilibrium. This is similar to the pacing equilibrium notion where the tie-breaking rule is endogenous to the equilibrium \cite{conitzer2022pacing}.
 }.  

\end{claim} 
\begin{proof}{Proof.}
 We show that in a first-price auction, the optimal bidding strategy is to bid on queries with a value-to-price ratio above a certain threshold.  To prove this, we assume that all advertisers' bidding profiles are given. Since the auction is first-price, advertiser $a$ can win each query $q$ by fixing small enough $\epsilon>0$ and paying $\max_{a'\neq a}(b_{a'}(q))+\epsilon$. So, let $p_a(q)=\max_{a'\neq a}(b_{a'}(q))$, be the price of query $q$.  Since we have assumed that the value functions of all advertisers are integrable (i.e., there are no measure zero sets of queries with a high value), in the optimal strategy, $p_a$ is also integrable since it is suboptimal for any advertiser to bid positive (and hence have a positive cost) on a measure zero set of queries.     

First, consider a budget-constrained advertiser. The main idea is that since the prices are integrable, the advertiser's problem is similar to a continuous knapsack problem. In a continuous knapsack problem, it is well known that the optimal strategy is to choose queries with the highest value-to-cost ratio \cite{goodrich2001algorithm}. Therefore, there must exist a threshold, denoted as $\mu$, such that the optimal strategy is to bid on queries with a value-to-price ratio of at least $\mu$. So if we let $\mu_a=\frac{1}{\mu}$, then advertiser $a$ bids on any query with $\mu_av_a(q)\geq p_a(q)$.

We prove it formally by contradiction.
 Assume to the contrary, that there exist non-zero measure sets $X,Y\subset Q$ such that for all $x\in X$ and $y\in Y$, the fractional value of $x$ is less than the fractional value of $y$, i.e., $\frac{v_a(x)}{p_a(q_x)}< \frac{v_a(y)}{p_a(y)}$, and in the optimal solution advertiser $a$ gets all the queries in $X$ and no query in $Y$. 
 However, we show that by swapping queries in $X$ with queries in $Y$ with the same price, the advertiser can still satisfy its budget constraint while increasing its value.

To prove this, fix $0<\alpha<\min(\int_Xp_a(q)dq,\int_Yp_a(q)dq)$. Since the Lebesgue measure is atomless, there exists subsets $X'\subseteq X$ and $Y'\subseteq Y$ such that $\alpha=\int_{X'}p_a(q)dq=\int_{Y'}p_a(q)dq$. Since the value per cost of queries in $Y$ is higher than queries in $X$, swapping queries of $X'$ with $Y'$ increases the value of the new sets while the cost does not change. Therefore, the initial solution cannot be optimal.

A similar argument holds for tCPA-constrained advertisers. Swapping queries in $X'$ with $Y'$ does not change the cost and increases the upper bound of the tCPA constraint, resulting in a feasible solution with a higher value. Therefore, the optimal bidding strategy for tCPA constraint is also $b_a(q)$ as defined in the statement of the claim.\end{proof}

\subsection*{Step 2: Equilibrium Characterization}

The previous step showed that the optimal bidding strategy is to bid on queries with a value-to-price ratio above a certain threshold. Thus, we need to track one variable per auto-bidder to find the sub-game equilibrium.

In what follows, we focus on the case of finding the variables when there are only two advertisers in the game. This characterization of equilibrium gives an implicit equation for deriving the equilibrium bidding strategy, which makes the problem tractable in our continuous-query model.\footnote{Notice that for the discrete-query model finding equilibrium is PPAD hard \cite{filos2021complexity}}. 

From Claim~\ref{claim: first price best response} we observe that the ratio of bidding multipliers is key to determining the set of queries that each advertiser wins. To map the space of queries to the bidding space, we introduce the function $h(q)=\frac{v_1(q)}{v_2(q)}$. Hence, for high values of $h$, the probability that advertiser 1 wins the query increases. Also, notice that without loss of generality, we can reorder the queries on $[0,1]$ so that $h$ is non-decreasing. 

In what follows, we further assume that $h$ is increasing on $[0,1]$. This implies that $h$ is invertible and also differentiable almost everywhere on $[0,1]$. With these assumptions in place, we can now state the following lemma to connect the sub-game equilibrium to the ratio of advertisers' values.

\begin{lemma}\label{lm: equilibrium first price} [Subgame Equilibrium in FPA]
Given two budget-constrained auto-bidders with budget $B_1$ and $B_2$,  let $\mu_1$ and $\mu_2$ be as defined in Claim~\ref{claim: first price best response} for auto-bidding with FPA. Also assume that  $h(q)=\frac{v_1(q)}{v_2(q)}$ as defined above is strictly monotone. Then $\mu_1=\frac{B_2}{ E[z\mathbbm{1}(z\geq r)]}$ and $\mu_2=\mu_1r$, where $r$ is the solution of the following implicit function,
\begin{equation}\label{eq: budget eq-first price}
    \frac{rE[\mathbbm{1}[z\geq r)]}{E[z\mathbbm{1}(z\leq {r})]}=\frac{B_1}{B_2}.
\end{equation}
Here, $E[\cdot]$ is defined as
$E[P(z)]=\int_0^\infty P(z)f(z)dz,$ where  $f(z)=\frac{v_2(h^{-1}(z))}{h'(h^{-1}(z))}$ wherever $h'$ is defined, and it is zero otherwise. 

Also, for two tCPA auto-bidders with targets $T_1$ and $T_2$, we have  $\mu_1=\frac{T_1 E[\mathbbm{1}(z\leq r)]}{ E[\mathbbm{1}(z\geq r)]}$ and $\mu_2=\mu_1r$, where $r$ is the answer of the following implicit function,
\begin{equation}\label{eq: tcpa eq-first price}
    \frac{rE[\mathbbm{1}(z\geq r)]}{E[z\mathbbm{1}(z\geq {r})]}\frac{E[\mathbbm{1}(z\leq r)]}{E[z\mathbbm{1}(z\leq {r})]}=\frac{T_1}{T_2}.
\end{equation}
\end{lemma}
\begin{remark}
{
The function $f$ intuitively represents the expected \emph{value} of the queries that advertiser 2 can win 
as well as the \emph{density} of the queries that advertiser 1 can win. Also, the variable $r$ shows the cut-off on how the queries are divided between the two advertisers.  
In the proof, we will see that the advertisers' value at equilibrium is computed with respect to $f$: Advertiser 1's overall value is $\int_r^\infty zf(z)dz$ and advertiser 2's overall value is $\int_0^r f(z)dz$. }
\end{remark}

\begin{proof}{Proof.}
First, consider budget constraint auto-bidders.
Given Claim~\ref{claim: first price best response}, in equilibrium price of query $q$ is $\min(\mu_1v_1(q),\mu_2v_2(q))$. Therefore, the budget constraints become:
\[B_1=\int_0^1  {\mu_2v_2(q)}\mathbbm{1}({\mu_2v_2(q)}\leq \mu_1v_1(q)) dq,\]
\[B_2=\int_0^1 {\mu_1v_1(q)} \mathbbm{1}({\mu_2v_2(q)}\geq \mu_1v_1(q))  dq.\]
With a change of variable from $q$ to $z=h(q)$ and letting $r=\frac{\mu_2}{\mu_1}$, we have:
\[B_1=\int_{r}^\infty{\mu_2v_2(h^{-1}(z))}\frac{dh^{-1}(z)}{dz}dz \]
\[B_2=\int_0^{r}{\mu_1v_1(h^{-1}(z))}\frac{dh^{-1}(z)}{dz}dz. \]
Observe that $v_1(h^{-1}(z))=zv_2(h^{-1}(z))$, then if we let $f(z)=v_2(h^{-1})(h^{-1})'=\frac{v_2(h^{-1}(z))}{h'(h^{-1}(z))}$, the constraints become
\begin{equation}\label{eq: B1 constraint}
    B_1=\mu_2\int_{r}^\infty f(z)dz, 
\end{equation}
\begin{equation}\label{eq: B2 constraint}
    B_2=\mu_1\int_0^{r}z f(z) dz.
\end{equation}
We obtain Equation~\eqref{eq: budget eq-first price} by diving both sides of Equation~\eqref{eq: B1 constraint} by the respective both sides of Equation~\eqref{eq: B2 constraint}.

Now, consider two tCPA-constrained auto-bidders. Similar to the budget-constrained auto-bidders, we can write
\[T_1\int_0^1  {v_1(q)}\mathbbm{1}({\mu_2v_2(q)}\leq \mu_1v_1(q)) dq=\int_0^1  {\mu_2v_2(q)}\mathbbm{1}({\mu_2v_2(q)}\leq \mu_1v_1(q)) dq\]
\[T_2\int_0^1  {v_2(q)}\mathbbm{1}({\mu_2v_2(q)}\geq \mu_1v_1(q)) dq=\int_0^1  {\mu_1v_1(q)}\mathbbm{1}({\mu_2v_2(q)}\geq \mu_1v_1(q)) dq\]
The same way of changing variables leads to the following:
\[
    \frac{T_1}{T_2}\frac{\int_r^{\infty}xf(x)dx}{\int_{0}^r f(x)}=\frac{r\int_r^{\infty}f(x)dx}{\int_{0}^r x f(x)dx}.
\]
This finishes the proof of the lemma.
\end{proof}

The previous theorem immediately implies that any example of valuation functions that is non-AIC for budget-advertisers will also be non-AIC for tCPA-advertisers.
\begin{corollary}\label{cor: tCPA implies budget}
If auto-bidding with FPA and two budget-advertisers is not AIC, then auto-bidding with the same set of queries and two tCPA-advertisers is also not AIC.
\end{corollary}
\begin{proof}{Proof.}
Recall that advertiser $1$ wins all queries with $h(q)\geq r$. So, the value accrued by advertiser 1 is decreasing in $r$. So, if an instance of auto-bidding with tCPA-constrained advertisers is not AIC for advertiser 1,  then the corresponding function $\frac{r\int_r^{\infty}f(x)dx}{\int_{0}^r x f(x)dx}\frac{\int_{0}^r f(x)}{\int_r^{\infty}xf(x)dx}$ (same as \eqref{eq: tcpa eq-first price}) must be increasing for some $r'$. 

On the other hand, recall that $\frac{r\int_{r}^\infty f(x)dx}{\int_{0}^r xf(x)dx}$ is the ratio for budget-constrained bidders equilibrium as in \eqref{eq: budget eq-first price}. The additional multiplier in the equilibrium equation of tCPA constraint advertiser in \eqref{eq: tcpa eq-first price} is  $\frac{\int_0^{r}f(x)dx}{\int_{r}^\infty xf(x)}$ which is increasing in $r$. So, if the auto-bidding for budget-constrained bidders is not  AIC and hence the corresponding ratio is increasing for some $r'$, it should be increasing for the tCPA-constrained advertisers as well, which proves the claim. 
\end{proof}

\subsection*{Step 3: Designing a non AIC instance}
The characterization of equilibrium from Step 2 leads us to construct an instance where advertisers have the incentive to misreport their constraints.  The idea behind the proof is that the value accrued by the advertiser $1$ is decreasing in $r$ (as found in Lemma~\ref{lm: equilibrium first price}). Then to find a counter-example, it will be enough to find an instance of valuation functions such that the equilibrium equation \eqref{eq: budget eq-first price} is non-monotone in $r$.
\begin{proof}{Proof of Theorem~\ref{thm: first price general}.}
We construct an instance with two budget-constrained advertisers. By Corollary~\ref{cor: tCPA implies budget}, the same instance would work for tCPA-constrained advertisers.
To prove the theorem, we will find valuation functions $v_1$ and $v_2$ and budgets $B_1$ and $B_2$ such that the value accrued by advertiser 1 decreases when their budget increases.

Define $g(r)=\frac{\int_0^{r}xf(x)dx}{r\int_{r}^\infty f(x)dx}$. 
By Lemma~\ref{lm: equilibrium first price}, one can find 
 the equilibrium by solving the equation $g(r)=\frac{B_2}{B_1}$. Recall that advertiser $1$ wins all queries with $\frac{v_1(q)}{v_2(q)}\geq r$. So, the total value of queries accrued by advertiser 1 is decreasing in $r$. Hence, it is enough to construct a non-AIC example to find a function $f$ such that $g$ is non-monotone in $r$.

 A possible such non-monotone function $g$ is
 \begin{align}\label{eq: g}
     g(r)=\frac{(r-1)^3+3}{cr}-1,
 \end{align}
 where $c$ is chosen such that $\min_{r\geq 0} g(r)=0$, i.e., $c=\min\frac{(r-1)^3+3}{r}\approx1.95105$. To see why $g$ is non-monotone, observe that $g(r)$ is  decreasing for $r\leq 1.8$, because $g'(r)=\frac{2r^3-3r^2-2}{cr^2}$ is negative for $r\leq 1.8$, and then increasing for $r\geq 1.81$. 
 
 We claim the function $f$ defined as in,
 \begin{equation}\label{eq: f}
     f(r)=3c(r-1)^2\frac{e^{\int_0^r\frac{c}{(r-1)^3+3}dx}}{((r-1)^3+3)^2},
 \end{equation}
 would result in the function $g$ in \eqref{eq: g}. To see why this claim is enough to finish the proof, note that there are many ways to choose the value functions of advertisers to derive t$f$ as in \eqref{eq: f}.
One possible  way is to define $v_1,v_2:[0,1]\to\mathbb R$ as $v_2(q)=f(\tan(q))/(\tan(q)^2+1)$ and $v_1(q)=\tan(q)v_2(q)$ (see Fig.~\ref{fig:example non AIC}).
\begin{figure}%
    \centering
    \subfloat[\centering Valuation function of two advertisers.]{{\includegraphics[width=6cm]{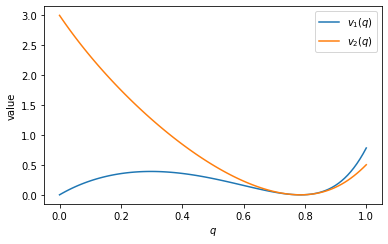} }}%
    \qquad
    \subfloat[\centering Finding the equilibrium using \eqref{eq: budget eq-first price}.]{{\includegraphics[width=6cm]{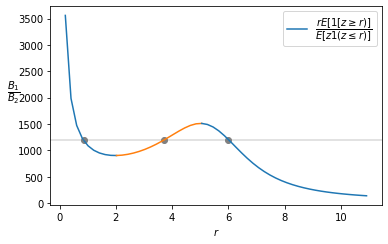} }}%
    \caption{An example of two advertisers such that FPA is not AIC (proof of Theorem~\ref{thm: first price general}). When $\frac{B_1}{B_2}=1200$, there are three values for $r$ (see the right panel) that lead to equilibrium, and one (orange) leads to non-AIC equilibrium.}%
    \label{fig:example non AIC}%
\end{figure}

So it remains to prove that choosing $f$ as in \eqref{eq: f} would result in $g$ as defined in \eqref{eq: g}. To derive $f$ from $g$, first we simplify $g$ using
 integration by part,
\begin{align*}
    g(r)&=\frac{\int_0^{r}xf(x)dx}{r\int_{r}^\infty f(x)dx}\\
    &=\frac{r\int_0^{r}f(x)dx-\int_0^r\int_0^xf(y)dydx}{r\int_{r}^\infty f(x)dx}\\
    &=\frac{r\int_0^{\infty}f(x)dx-\int_0^r\int_0^xf(y)dydx}{r\int_{r}^\infty f(x)dx}-1,
\end{align*}
Assuming that $\int_0^\infty f(x)$ is finite, the above equations lead to the following
\begin{equation}\label{eq: rgr+r}
    rg(r)+r=\frac{\int_0^r\int_x^\infty f(y)dydx}{\int_r^\infty f(x)dx}.
\end{equation}
Therefore, by integrating the inverse of both  sides,
\[\log(\int_0^r\int_x^\infty f(y)dy dx)=C+\int_0^r\frac{1}{xg(x)+x}dx,\]
and by raising the exponent
\[
    \int_0^r\int_x^\infty f(y)dydx=Ke^{\int_0^r\frac{1}{xg(x)+x}dx}.
\]
 for some constants $C$ and  $K>0$. Then by differentiating both sides with respect to $x$,
\[\int_r^\infty f(x)dx=\frac{K}{rg(r)+r}e^{\int_0^r\frac{1}{xg(x)+x}dx}.\]
Note that for any choice of $K\geq 0$, dividing the last two equations will result in \eqref{eq: rgr+r}. So, without loss of generality, we can assume $K=1$. 
By differentiating again,  we can derive $f$ as a function of $g$:
\[f(r)=\frac{(g'(r)r+g(r))}{(rg(r)+r)^2}e^{\int_0^r\frac{1}{xg(x)+x}dx}.\]
We need $g'(r)r+g(r)\geq 0$ to ensure that 
 $f(r)\geq 0$ for all $r$. This holds for $g$ as in \eqref{eq: g}. Finally, by substituting $g$ as in \eqref{eq: g}, we will derive $f$ as in \eqref{eq: f}.  
\end{proof}
\begin{remark}
    Note that the above proof shows that for values of $r$ such that there exists an equilibrium that is not AIC, there always exists a second monotone equilibrium. This follows from the fact that the function $g(r)$ tends to infinity as $r\to\infty$, so $g$ must be increasing for some large enough $r$.  
\end{remark}
Before moving on to finding conditions for incentive compatibility, we also note that the above's characterization implies the existence of equilibrium for auto-bidding with any pairs of advertisers.

\begin{proposition}
Given auto-bidding satisfying the conditions of Lemma~\ref{lm: equilibrium first price}, the equilibrium for all pairs of budgets or all pairs of tCPA-constrained advertisers always exists.
\end{proposition}
\begin{proof}{Proof.}
Recall that the equilibrium exists if there exists an $r$ such that 
$$\frac{B_2}{B_1}=\frac{\int_0^{r}xf(x)dx}{r\int_{r}^\infty f(x)dx}$$ has a solution for any value of $\frac{B_2}{B_1}$.
Note that the right-hand side ($\frac{\int_0^{r}xf(x)dx}{r\int_{r}^\infty f(x)dx}$) is positive for any $r>0$, and it continuously grows to infinity as $r\to\infty$. So, to make sure that every value of $B_2/B_1$ is covered, we need to check whether at $r=0$ the ratio becomes zero. By L'Hopital rule,
$\lim_{z\to 0}\frac{zf(z)}{\int_z^\infty f(x)dx-zf(z)}=0$, which is  as desired.

For tCPA-constrained advertiser, the second ratio $\frac{r\int_0^rf(x)dx}{\int_r^\infty xf(x)dx}$ always converges to $0$, so the equilibrium in this case always exists.
\end{proof}

\subsection{Sufficient Conditions for Incentive Compatibility}\label{sec: sufficient conditions}
In this section, we show that the lack of ACI happens in cases where advertisers' valuations have unusual properties. More precisely, the main result of the section is to characterize sufficient conditions on the advertiser's valuations so that FPA is AIC when there are two advertisers in the auction.

For this goal, we recall the function $f(z)=\frac{v_2(h^{-1}(z))}{h'(h^{-1}(z))}$ where $h(q)=\frac{v_1(q)}{v_2(q)}$ defined in Section~\ref{sec: FPA is not AIC}. As shown in Lemma\ref{lm: equilibrium first price}, function $f$ behaves as a value of the queries advertiser $2$ gets and the density of queries that advertiser $1$ gets.

 \begin{lemma}\label{lem: AIC condition f}
Consider that there are two advertisers, and they can either be budget-advertisers or tCPA-advertisers. Also, suppose that auto-bidder with FPA uses the optimal bidding strategy in Claim~\ref{claim: first price best response}. Then a sufficient condition for FPA to be AIC is that $f$ has a monotone hazard rate, i.e., $\frac{f(r)}{\int_r^\infty f(x) dx}$ is non-decreasing in $r$.
\end{lemma}
\begin{proof}{Proof.}
   Following the proof Theorem~\ref{thm: first price general}, if $g(r)=\frac{\int_0^r\int_x^\infty f(y)dydx}{r\int_r^\infty f(x)dx} $ is non-decreasing in $r$ then the equilibrium is  AIC. The equivalent sufficient conditions obtained by imposing the inequality $g'(r)\geq0$ is that for all $r\geq 0$,
   \begin{equation}\label{eq: mhr}
       r\big(\int_r^\infty f(x) dx \big)^2\geq \big(\int_0^r\int_x^\infty f(y)dydx\big)\big(\int_r^\infty f(x) dx-rf(r) \big).
   \end{equation}
   If $\int_r^\infty f(x) dx\leq rf(r) $, then the above's inequality obviously holds. So, we can assume that for some $r>0$,  $\int_r^\infty f(x) dx> rf(r)$. Since $\frac{f(z)}{\int_z^\infty f(x) dx}$ is non-decreasing in $z$, we must have that for all $r'\leq r$,  $\frac{f(r')}{\int_{r'}^\infty f(x) dx}\leq  \frac{f(r)}{\int_r^\infty f(x) dx} \leq \frac{1}{r}\leq \frac{1}{r'}$. On the other hand by taking the derivative of $\frac{f(z)}{\int_z^\infty f(x) dx}$, we must have that $f'(z)\int_z^\infty f(x)dx+f(z)^2\geq 0$. By considering two cases on the sign of $f'$, for $z\leq r$ we must have, $f(z)\big(zf'(z)+f(z))\geq 0$, and hence $(zf(z))'\geq0$ for all $z\leq r$. Therefore, $zf(z)$ is non-decreasing for $z\leq r$.  

On the other hand,
\begin{align*}
    \int_0^r\int_x^\infty f(y)dydx&=\int_0^r\int_x^r f(y)dydx+\int_0^r\int_r^\infty f(y)dydx\\
    &=r\int_0^rf(x)dx-\int_0^r\int_0^x f(y)dydx+r\int_r^\infty f(x)dx\\
    &=\int_0^rxf(x)dx+r\int_r^\infty f(x)dx,
\end{align*}
where the second equation is by integration by part. Then by applying monotonicity of $z(f)z$ for $z\leq r$ we have
$ \int_0^r\int_x^\infty f(y)dydx\leq r^2f(r)+r\int_r^\infty f(x)dx$. So, to prove \eqref{eq: mhr} it is enough to show that
\begin{align*}
       \left(\int_r^\infty f(x) dx \right)^2\geq \left(rf(r)+\int_r^\infty f(x)dx\right)\left(\int_r^\infty f(x) dx-rf(r) \right),
\end{align*}
which holds, since the right-hand side is equal to $\big(\int_r^\infty f(x) dx \big)^2-(rf(r))^2$ strictly less than the left-hand side.
\end{proof}

While the condition on $f$ has intuitive properties when seen as a density, it has the unappealing properties to be too abstract in terms of the conditions on the advertisers' valuation. The following result provides sufficient conditions on value functions $v_1$ and $v_2$ that make $f$  be monotone hazard rate, and hence, FPA to be AIC.

\begin{theorem}\label{thm: sufficient condition-first price}
Consider two advertisers that are either budget-advertisers or tCPA-advertisers. Assume that $h(q)=\frac{v_1(q)}{v_2(q)}$ is increasing concave function and that $v_2$ is non-decreasing. Then, the equilibrium in FPA auto-bidding with bidding strategy as in Claim~\ref{claim: first price best response} is AIC. 
\end{theorem}
\begin{proof}{Proof.}
    Note that when $f$  is non-decreasing, it also has a monotone hazard rate. Now, when $h$ is concave, $\frac{1}{h'}$ is a non-decreasing function, and since $v_2$ is also non-decreasing, then $f$ is also non-decreasing.
\end{proof}

\begin{remark}[AIC test for general valuations]
    While the conditions outlined by Theorem~\ref{thm: sufficient condition-first price} for valuation functions might initially appear restrictive, it's important to highlight that our characterization in Lemma~\ref{lm: equilibrium first price}  serves as a valuable test to determine AIC condition in a broader range of scenarios. As long as the solution $r$ is monotone with respect to the ratio of budgets and targets, then the auction is AIC for the given valuations.
For market structures where the equilibrium solution adheres to the monotonicity condition, any resulting equilibrium inherently satisfies the AIC property. The sufficient conditions elucidated by Theorem~\ref{thm: sufficient condition-first price} is just one example of such monotonicity conditions. 
As we delve further, Section~\ref{sec: truthful} will uncover a parallel characterization applicable to a wider spectrum of truthful auctions, reinforcing the generality of our approach.
\end{remark}

\subsection{FPA with uniform bidding}\label{sec: pacing}
The previous section shows that when auto-bidders have full flexibility in the bidding strategy, FPA is not AIC. However, non-uniform bidding is more complex to implement, and auto-bidders may be constrained to use simpler uniform bidding policies (aka pacing bidding). In this context, the main result of the section is Theorem~\ref{thm:first price pacing} that shows that when restricted to uniform bidding policies, FPA is AIC. Note that here, we assume a simple model where advertisers do not split campaigns. So, FPA with uniform bidding is AIC, but it could bring up other incentives for advertisers when it is implemented.

\begin{definition}[Uniform bidding equilibrium]
A uniform bidding equilibrium for the auto-bidders subgame corresponds to bid multipliers $\mu_1,\ldots,\mu_N$ such that every auto-bidder $a$ chooses the uniform bidding policy $\mu_a$ that maximizes Problem~\eqref{eq: auto-bidder agent-descrete} when restricted to uniform bidding policies with the requirement that if advertiser $a$'s constraints of type \ref{eq: constraint} are not tight then $\mu_a$ gets its maximum possible value.\footnote{ When valuations are strictly positive for all queries $q\in [0,1]$, we can easily show that bid multipliers have to be bounded in equilibrium. When this is not the case, we set a cap sufficiently high to avoid bid multipliers going to infinity.}
\end{definition}

The proof of Theorem~\ref{thm:first price pacing} is based on the main results of \cite{conitzer2022pacing}. The authors proved that the uniform-bidding equilibrium is unique, and in equilibrium, the multiplier of each advertiser is the maximum multiplier over all feasible uniform bidding strategies. Their result is for budget-constrained advertisers, and we extend it to include tCPA-constrained advertisers. The proof is deferred to Appendix \ref{apendix: pacing}.
\begin{lemma}[Extension of Theorem 1 in \cite{conitzer2022pacing}]\label{lm: conitzer}
Given an instance of Auto-bidding with general constraints as in \eqref{eq: constraint}, there is a unique uniform bidding equilibrium, and the bid multipliers of all advertisers are maximal among all feasible uniform bidding profiles.
\end{lemma}

Now, we are ready to prove Theorem \ref{thm:first price pacing}.
\begin{proof}{Proof of Theorem \ref{thm:first price pacing}}
Assume that advertiser 1 increases their budget or their target CPA. Then the original uniform bidding is still feasible for all advertisers. Further, by Lemma~\ref{lm: conitzer}, the equilibrium pacing of all advertisers is maximal among all feasible pacing multipliers. So, the pacing of all advertisers will either increase or remain the same. But the constraints of all advertisers except 1 are either binding, or their multiplier has attained its maximum value by the definition of pacing equilibrium. Therefore, the set of queries they end up with should be a subset of their original ones since the price of all queries will either increase or remain the same. So, it is only advertiser 1 that can win more queries.
\end{proof}

\begin{remark}
\cite{conitzer2022pacing} show monotonicity properties of budgets in FPA with uniform bidding equilibrium for the revenue and welfare. Instead, in our work, we focus on monotonicity for each advertiser.  
\end{remark}

\section{Truthful Auctions}\label{sec: truthful parent section}

This section studies auto-bidding incentive compatibility for the case where the per-query auction is truthful. 

A truthful auction is an auction where the optimal bidding strategy for a profit-maximizing agent is to bid its value. An important example of a truthful auction is the Second Price Auction. As we showed in the three-queries example of the introduction, SPA is not AIC. In this section, we show that the previous example generalizes, in our continuous-query model, to any (randomized) truthful auctions so long as the auction is scale-invariant and symmetric (see Assumption~\ref{ass: scalibility-symmetry} below for details). As part of our proof technique, we obtain an auction equivalence result that is interesting on its own: in the continuous query-model, SPA and FPA have the same outcome.\footnote{ It is well-known that in the discrete-query model, FPA and SPA are not auction equivalent in the presence of auto-bidders.} 

For the remaining of section we assume all truthful auctions satisfy the following property.

\begin{assume}\label{ass: scalibility-symmetry}
Let $(x_a(\mathbf b))_{a\in A}$ be the allocation rule in a truthful auction given bids $\mathbf b = (b_a)_{a\in A}$.
We assume that the allocation rule satisfies the following properties.
\begin{enumerate}\label{assume: fair and truthful}
\item The auction always allocates: $\sum_{a\in A}x_a(\mathbf b)=1$
    \item Scale-invariance: For any constant $c>0$ and any advertiser $a\in A$,
    $x_a(\mathbf b)=x_a(c\mathbf b) $.
    \item Symmetry: For any pair of advertisers $a,a'\in A$ and bids $b,b',\vec b_{-\{a,a'\}}= (b)_{a\in A\setminus \{a,a'\}}$ we have that
    $$ x_a(b_a = b, b_{a'} = b', \vec b_{-\{a,a'\}}) := x_{a'}(b_a = b', b_{a'} = b, \vec b_{-\{a,a'\}}).$$
\end{enumerate}
\end{assume}
\begin{remark}
    Observe that SPA satisfies Assumption~\ref{ass: scalibility-symmetry}.
\end{remark}

From the seminal result of \citet{myerson81} we obtain a tractable characterization of truthful auctions, which we use in our proof. 

\begin{lemma}[Truthful auctions \citep{myerson81}]\label{asum :truthful}
Let $(x_a(\vec b), p_a(\vec b))_{a\in A}$ the allocation and pricing rule for an auction given bids $\vec b = (b_a)_{a\in A}$. The auction rule is truthful if and only if
\begin{enumerate}
    \item Allocation rule is non-decreasing on the bid: For each bidder $a\in A$ and any $b_a'\geq b_a$, we have that  $$x_a(b'_a,\vec b_{-a})\geq x_a(b_a,\vec b_{-a}).$$
    \item Pricing follows Myerson's formulae:
    \[p_a(\vec b)=b_a\cdot x_a(\vec b)-\int_0^{b_a}x_a(z,\vec b_{-a})dz.\]
\end{enumerate}
\end{lemma}

A second appealing property of truthful actions is that the optimal bidding strategy for auto-bidders is simpler: in the discrete-query model, uniform bidding strategy is almost optimal and can differ from optimal by, at most, the value of two queries \citep{Aggarwal19}. We revisit this result in our continuous-query model and show that a uniform bidding policy is optimal for truthful auctions.

\begin{claim}\label{claim: uniform bididng}
In the continuous-query model, if the per-query auction is truthful, uniform bidding is an optimal strategy for each auto-bidder.
\end{claim} 
\begin{proof}{Proof.} We use Theorem 1 \cite{Aggarwal19}. Pick some small $\delta>0$ and divide the interval $[0,1]$ into subintervals of length $\delta$. Let each subinterval $I$ be a discrete query with value functions $v_j(I)=\int_I v_j(q)dq$. Then Theorem 1 \cite{Aggarwal19} implies that uniform bidding differs from optimal by at most two queries. So, the difference from optimal is bounded by $2\max_j\max_{|I|\leq\delta} v_j(I)$. Now, since the valuation functions are atomless (i.e., the value of a query is $dq$), by letting $\delta$ to $0$, the error of uniform bidding in the continuous case also goes to zero.
\end{proof}

\subsection{SPA in the Continuous-Query Model}\label{sec: sec
ond price general}
We generalize the discrete example of second price auction in Theorem~\ref{thm: second price Warm up} to the continuous set of queries model showing that SPA is not AIC. The key step shows that for the continuous-query model, there is an auction equivalence result between the first and second price. 
\begin{theorem}\label{thm: equivalent first and second}[Auction Equivalence Result]
Suppose that auto-bidder uses a uniform bid strategy for SPA and similarly, uses the simple bidding strategy defined in Claim~\ref{claim: first price best response} for FPA. 
Then, in any subgame equilibrium, the outcome of the auctions (allocations and pricing) on SPA is the same as in FPA.
\end{theorem}

This result immediately implies that all the results for FPA in Section~\ref{sec: first price genral} hold for SPA as well.
\begin{theorem}
Suppose that there are at least two budget-advertisers or two tCPA-advertisers; then, even for the continuous-query model, SPA is not AIC.
\end{theorem}

Similarly to the FPA case, we can characterize the equilibrium for the two-advertiser case and derive sufficient conditions on advertisers' valuation functions so that SPA is AIC.

\begin{theorem}\label{thm: continuous-second price}
Given two advertisers,  let $\mu_1$ and $\mu_2$ be the bidding multipliers in equilibrium for the subgame of the auto-bidders. Also assume that  $h(q)=\frac{v_1(q)}{v_2(q)}$ is increasing. Then 
\begin{enumerate}
    \item If the advertisers are budget-constrained with budget $B_1$ and $B_2$, then $\mu_1=\frac{B_2}{ E[z\mathbbm{1}(z\geq r)]}$ and $\mu_2=\mu_1r$, where $r$ is the answer of the following implicit function,
\begin{equation*}
    \frac{rE[\mathbbm{1}[z\geq r)]}{E[z\mathbbm{1}(z\leq {r})]}=\frac{B_1}{B_2}.
\end{equation*}
Here, $E[.]$ is defined as
$E[P(z)]=\int_0^\infty P(z)f(z)dz,$ where  $f(z)=\frac{v_2(h^{-1}(z))}{h'(h^{-1}(z))}$ wherever $h'$ is defined, and it is zero otherwise. 
\item If the advertisers are tCPA-constrained with targets $T_1$ and $T_2$,   we have  $\mu_1=\frac{T_1 E[\mathbbm{1}(z\leq r)]}{ E[\mathbbm{1}(z\geq r)]}$ and $\mu_2=\mu_1r$, where $r$ is the answer of the following implicit function,
\[
    \frac{rE[\mathbbm{1}(z\geq r)]}{E[z\mathbbm{1}(z\geq {r})]}\frac{E[\mathbbm{1}(z\leq r)]}{E[z\mathbbm{1}(z\leq {r})]}=\frac{T_1}{T_2}.
\]

\item  If further, $v_2$  is non-decreasing in $q$, and $h$ is concave, and advertisers are either both budget-constrained or tCPA-constrained, then SPA is AIC. 
\end{enumerate}

\end{theorem}

We now demonstrate the auction equivalence between FPA and SPA. 

\begin{proof}{Proof of Theorem~\ref{thm: equivalent first and second}.}
Note that the optimal strategy for a second-price auction is uniform bidding with respect to the true value of the query by Claim~\ref{claim: uniform bididng}. Also, Claim~\ref{claim: first price best response} implies that the cost obtained by each advertiser in the first-price auction in the continuous model also depends on the pacing multipliers of the other advertiser. This claim immediately suggests the equivalent between the optimal bidding strategies of first and second-price auctions. So, the optimal strategy for both auctions will be the same; therefore, the resulting allocation and pricing will also be the same. Hence,
the same allocation and pricing will be a pure equilibrium under both auctions. 
\end{proof}

\subsection{Truthful Auctions Beyond Second-Price}\label{sec: truthful}
We now present the main result of the section. We show that a general truthful auction (with possibly random allocation) is not AIC.
\begin{theorem}\label{thm: truthful}
Consider a truthful auction $(\vec x, \vec p)$ satisfying Assumption~\ref{ass: scalibility-symmetry}. If there are at least two budget-advertisers or two tCPA-advertisers, then the truthful auction is not AIC. 
\end{theorem}

The remainder of the section gives an overview of the proof of this theorem.  
Similar to the FPA and SPA case, we start by characterizing the equilibrium in the continuous case when there are two advertisers in the game.  
The proof relies on the observation that for auctions satisfying Assumption~\ref{ass: scalibility-symmetry}, the allocation probability is a function of the bids' ratios. So, again, similar to FPA and SPA, finding the equilibrium reduces to finding the ratio of bidding multipliers. 
Then to finish the proof of Theorem~\ref{thm: truthful}, instead of providing an explicit example where auto-bidding is non-AIC, we showed that the conditions needed for an auction's allocation probability to satisfy are impossible.

The following theorem finds an implicit equation for the best response.
We omit the proofs of the intermediary steps and deferred them to Appendix~\ref{appendix: proofs truthful}.

\begin{theorem}\label{thm: equilibrium myerson}
Consider a truthful auction $(\vec x, \vec p)$ satisfying Assumption~\ref{ass: scalibility-symmetry} and assume that there are either two budget-advertisers or two tCPA-advertisers. 
Let $\mu_1$ and $\mu_2$ be the bidding multipliers used by the auto-bidders in the sub-game equilibrium. Further, assume that  $h(q)=\frac{v_1(q)}{v_2(q)}$ is increasing. Then 
\begin{enumerate}
    \item If the advertisers are budget-constrained with budget $B_1$ and $B_2$, then $\mu_1=\frac{B_1}{ E[p_1(rz,1)]}$ and $\mu_2=r\mu_1$, where $r$ is the answer of the following implicit function,
\[\frac{E[rp_1(\frac{z}{r},1)]}{E[zp_1(\frac{r}{z},1)]}=\frac{B_1}{B_2}.\]
Here, $E[.]$ is defined as
$E[P(z)]=\int_0^\infty P(z)f(z)dz,$ where  $f(z)=\frac{v_2(h^{-1}(z))}{h'(h^{-1}(z))}$ wherever $h'$ is defined, and it is zero otherwise. 
\item If the advertisers are tCPA-constrained with targets $T_1$ and $T_2$,   we have  $\mu_1=\frac{T_1 E[zg(z/r)]}{ E[rp_1(z/r)]}$ and $\mu_2=\mu_1r$, where $r$ is the answer of the following implicit function,
\[\frac{E[x_1(\frac{r}{z},1)]}{E[zx_1(\frac{z}{r},1)]}\frac{E[rp_1(\frac{z}{r},1)]}{E[zp_1(\frac{r}{z},1)]}=\frac{T_1}{T_2}.\]
\end{enumerate}
\end{theorem}

Because allocation probability $x_1$ is a non-decreasing function, we can derive a similar result to the FPA case and show if an instance is not AIC for budget-advertisers then it is also not AIC for tCPA-advertisers.

\begin{proposition}\label{prop: budget stronger truthful}
If, for the two budget-constrained advertisers' case, the truthful auction is not AIC, then for the tCPA-constrained advertisers' case, the same auction is also not AIC.
\end{proposition}

Using the previous results, we are in a position to tackle the main theorem.

\begin{proof}{Proof of Theorem~\ref{thm: truthful}.}
We prove Theorem~\ref{thm: truthful} for budget-constrained advertisers, since Proposition~\ref{prop: budget stronger truthful} would derive it for tCPA constraint advertisers.
We use the implicit function theorem to find conditions on $p_1$ and $f$ to imply monotonicity in $r$. Let 
\[H(x,r)=\frac{\int_0^\infty rf(z)p_1(z/r,1)dz}{\int_0^\infty f(z)z p_1(r/z,1)dz}-x.\]
Then when advertiser $1$ increases the budget, the corresponding variable $x$ increases. So, if we want to check whether $r$ is a  non-decreasing function of $x$, we need $\frac{dr}{dx}$ to be non-negative. By the implicit function theorem,
\begin{align*}
    \frac{dr}{dx}=-\frac{\frac{\partial H}{\partial x}}{\frac{\partial H}{\partial r}}=\frac{1}{\frac{\partial H}{\partial r}}.
\end{align*}

So, assume to the contrary that $r$ is always non-decreasing in x, then $\frac{\partial H(x,r}{\partial r}\geq 0$. Define $p(x)=p_1(x,1)$. Then we have the following
\begin{align*}
E[\frac{d}{dr}rp(z/r)]E[zp(r/z)]\geq E[rp(z/r)]E[\frac{d}{dr}\Big(zp(r/z)\Big)].
\end{align*}
Then
\[\frac{\frac{d}{dr}E[rp(z/r)]}{E[rp(z/r)]}\geq \frac{\frac{d}{dr}E[zp(z/r)]}{E[zp(z/r)]}\]
By integrating both parts, we have that for any choice of $f$,
\[rE[p(z/r)]\geq E[z p(r/z)].\]
When the above inequality hold for any choice of $v_1$ and $v_2$, we claim that the following must hold  almost everywhere
\begin{equation}\label{eq: truthful-bnd}
p(b)\geq  bp(1/b).
\end{equation}
To see this, assume to the contrary that a measurable set $B$ exists such that \eqref{eq: truthful-bnd} does not hold for it. Let $qv_2(q)=v_1(q)$, therefore, $f(z)=v_2(z)$ can be any measurable function. So, we can define $f$ to have zero value everywhere except $X$ and have weight $1$ over $X$ to get a contradiction.

By substituting  variable with $y=1/b$ in \eqref{eq: truthful-bnd}, $p(1/b)db\geq  p(b)/b db$. Therefore, almost everywhere $p(b)=bp(1/b)$. By differentiating we have $p'(b)=p(1/b)-p'(1/x)/x$. On the other hand, as we will see in Appendix~\ref{appendix: proofs truthful} for any truthful auction satisfying Assumption~\ref{ass: scalibility-symmetry}, $p'(b)=p'(1/b)$.
Therefore, $p(b)=p'(b)(b+1)$. Solving it for $p$, we get that the only possible  AIC pricing must be of the form $p(b)=\alpha(b+1)$ for some $\alpha>0$.

Next, we will show there is no proper allocation probability satisfying the Assumption~\ref{ass: scalibility-symmetry} that would result in a pricing function $p$. It is not hard to see that by Myerson's pricing formulae,
$\frac{dx_1(b,1)}{db}=\frac{p'(b)}{b}.$
Therefore, we must have  $x_1'(b,1)=\alpha/b$, so $x_1(b,1)=c\log(b)+d$ for some constants $c>0$ and $d$. But $x_1$ cannot be a valid allocation rule since it will take negative values for small enough $b$.
\end{proof}

\bibliographystyle{ACM-Reference-Format}
\bibliography{references}


\begin{appendix}

\section{SPA with Discrete Queries (Omitted Proofs)}\label{appendix: tcpa second price warm up}

\begin{proof}{Proof of Theorem~\ref{thm: second price Warm up} (cont.)}We continue with the proof of Theorem~\ref{thm: second price Warm up}. Our goal is to establish the uniqueness of equilibrium with respect to the allocation of queries in the case where $B_1'=10$. Let us first define $h(q)=\frac{v_1(q)}{v_2(q)}$ as the ratio of values of each query for advertiser 1 over advertiser 2. Notably, we have the following order for $h(q)$: $h(q_1)>h(q_2)>h(q_3)>h(q_4)$.
In any equilibrium with uniform bidding, advertiser 1 wins queries in the order given by $h$, namely $q_1$, $q_2$, $q_3$, and finally $q_4$. Conversely, advertiser 2 wins in the reverse order, implying that if, for example, advertiser 2 wins $q_2$, it must have already won $q_3$ and $q_4$. We already show that there exists an equilibrium in that advertiser 1 wins $q_1, q_2,$ and $q_3$. 
With the above argument, to establish the uniqueness of equilibrium, we need to rule out four cases based on whether advertiser 1 wins nothing, or $q_1$, or $q_1$, and $q_2$, or all queries.

Case 1 (Advertiser 1 wins all queries):
To begin, we demonstrate that in every equilibrium, advertiser 2 wins at least one query. This can be observed by noting that the multiplier of advertiser 1 is at most $1$. As a result, the price of $q_4$ for advertiser 2 is at most $20$, which falls within the budget of advertiser 2. By setting their multiplier to $\mu_2=h(q_4)\mu_1 + \epsilon$, where $\epsilon$ is a small positive value, autobidder 2 can win only $q_4$ at a price of $\mu_1v_1(q_4)\leq 20<B_2$. Therefore, there is no equilibrium where advertiser 2 wins no queries.

Case 2 (Advertiser 1 wins nothing):
On the other hand, we now demonstrate that in any equilibrium, advertiser 1 wins at least $q_1$. The price of $q_1$ for advertiser 1 is at most $v_2(q_1)=1$, which falls within the budget $B_1'$ of advertiser 1. Thus, if in some equilibrium advertiser 1 does not win any query, their autobidder can set $\mu_1=\frac{\mu_2}{h(q_1)}-\epsilon$ for some small $\epsilon$ such that $2.1=h(q_1)>\frac{\mu_2}{\mu_1}>h(q_2)=2$. Consequently, advertiser 1 wins only $q_1$ with a price of at most $\mu_2v_2(q_1)\leq 1$, within the budget $B_1'$.

Case 3 (Advertiser 1 wins $q_1$):
Assume that there exists an equilibrium where advertiser one only wins $q_1$ (and advertiser 2 wins the rest of the queries). By allocation inequalities, we must have $\frac{\mu_2}{\mu_1}>h(q_2)=2$. Using $\mu_2\leq 1$, we get $\mu_1\leq .5$. However, in this case, the price of all queries for advertiser 2 is $92.1\mu_1<46.05$, so advertiser 2 has the incentive to increase $\mu_2$ to win all queries. This contradicts the assumption, and thus there is no equilibrium where advertiser 1 just wins $q_1$.

Case 4 (Advertiser 1 wins $q_1$ and $q_2$):
So it remains to rule out the case that advertiser 1 wins only $q_1$ and $q_2$.
Assume to the contrary that some equilibrium exists with bidding multipliers  $\tilde\mu_1$ and $\tilde\mu_2$  such that advertiser  1 gets  $q_1$ and $q_2$. Then
\[ \frac{\tilde \mu_1(v_1(q_2)+v_1(q_3)+v_1(q_4))}{B_2}>1\geq  \frac{\tilde \mu_2(v_2(q_1)+v_2(q_2))}{B_1'},\] 
the first inequality is because advertiser 2's multiplier is the best response, and the second comes from the budget constraint for advertiser 1.
Therefore,
\[ \frac{B_1'}{B_2}\frac{v_1(q_2)+v_1(q_3)+v_1(q_4)}{v_2(q_1)+v_2(q_2)}\geq \frac{\tilde \mu_2}{\tilde \mu_1}.\]
But $ \frac{B_1'}{B_2}\frac{v_1(q_2)+v_1(q_3)+v_1(q_4)}{v_2(q_1)+v_2(q_2)}=\frac{10}{49}\times \frac{90}{21}\approx .87$, which is smaller than $h(q_3)=1.2$. Thus $h(q_3)>\frac{\tilde\mu_2}{\tilde\mu_1}$. This contradicts allocation inequalities since we assumed advertiser 2 wins $q_3$ and $q_4$ given the multipliers $\tilde \mu_1$ and $\tilde\mu_2$. With this contradiction, we see that for $B'_1=10$ and $B_2=49$, the equilibrium is unique where advertiser 1 wins $q_1$, $q_2$, and $q_3$.

Next, we show a non AIC example for tCPA advertisers. Again consider two advertisers, and three queries, with values as in Table~\ref{tab:tcpa-second price}. Here, let the tCPA constraint of advertiser 1 be $T_1=0.4$ and for advertiser 2 be $T_2=0.7$.
Then again, we show that there exists a unique equilibrium in which advertiser 1 gets queries 1 and 2. 
\begin{table}[!htp]
    \centering
    \begin{tabular}{c|c|c|c}
         & $q_1$& $q_2$& $q_3$ \\\hline
        Advertiser 1 & 40 & 30 &20\\
        Advertiser 2 & 10 & 13 &100\\
    \end{tabular}
    \caption{SPA with two tCPA constraint advertisers is not AIC: The value of each query for each advertiser.}
    \label{tab:tcpa-second price}
\end{table}

First, to prove the existence, let $\mu_1=1.6$ and $\mu_2=1.2$. Then we show this is an equilibrium since the three following conditions hold:
\begin{enumerate}
    \item Allocation: advertiser 1 wins $q_1$ and $q_2$ since it has a higher bid on them
\[\frac{v_1(q_1)}{v_2(q_1)}\geq\frac{v_1(q_2)}{v_2(q_2)}\geq \frac{\mu_2}{\mu_1}=\frac{1.2}{1.6}\geq  \frac{v_1(q_3)}{v_2(q_3)}.\]
    \item tCPA constraints are satisfied:
    \[70=T_2v_2(q_3)\geq \mu_1 v_1(q_3)=32,\qquad\text{and }\]
    \[28=T_1(v_1(q_1)+v_1(q_2))\geq \mu_2(v_2(q_1)+v_2(q_2))=27.6.\]
\item Best response: non of the advertiser can win more queries if they increase their multiplier:
    \[79.1=T_2(v_2(q_3)+v_2(q_2))< \mu_1(v_1(q_3)+v_1(q_2))=80,\]
    \[86.1=T_2(v_2(q_3)+v_2(q_2)+v_2(q_1))< \mu_1(v_1(q_3)+v_1(q_2)+v_1(q_1))=144,\qquad\text{and}\]
    \[36=T_1(v_1(q_1)+v_1(q_2)+v_1(q_3))< \mu_2 (v_2(q_1)+v_2(q_2)+v_2(q_3))=147.6\]
\end{enumerate} 
Now, similar to the proof of the budget-constrained advertisers, we show the equilibrium is unique. Note that there's no equilibrium such that advertiser 1, gets all queries since the cost of all queries for advertiser 1 is at least $v_2(q_1)+v_2(q_2)+v_2(q_3)=123$, which is larger than $T_1(v_1(q_1)+v_1(q_2)+v_1(q_3))$. Similarly, advertiser 2 cannot get all queries since the tCPA constraint would not hold $v_1(q_1)+v_1(q_2)+v_1(q_3)>T_2(v_2(q_1)+v_2(q_2)+v_2(q_3))$. So, to prove the uniqueness of equilibrium, it remains to show that there's no equilibrium that advertiser 1 only gets query 1. 
To contradiction, assume such equilibrium exists with the corresponding multipliers $\tilde\mu_1$ and $\tilde\mu_2$. Then we must have 
\[ \frac{\tilde \mu_1(v_1(q_1)+v_1(q_2)+v_1(q_3))}{T_2(v_2(q_1)+v_2(q_2)+v_2(q_3))}>1\geq  \frac{\tilde \mu_2v_2(q_1)}{T_1v_1(q_1)},\] 
where the first inequality is because advertiser 2's multiplier is the best response, and the second inequality comes from the budget constraint for advertiser 1.
Therefore,
\[ \frac{T_1}{T_2}\frac{v_1(q_1)+v_1(q_2)+v_1(q_3)}{v_2(q_1)+v_2(q_2)+v_2(q_3)}\frac{v_1(q_1)}{v_2(q_1)}\geq \frac{\tilde \mu_2}{\tilde \mu_1},\]
But $\frac{v_1(q_2)}{v_2(q_2)}\geq \frac{T_1}{T_2}\frac{v_1(q_1)+v_1(q_2)+v_1(q_3)}{v_2(q_1)+v_2(q_2)+v_2(q_3)}\frac{v_1(q_1)}{v_2(q_1)}\approx 1.67$, and thus $\frac{v_1(q_2)}{v_2(q_2)}>\frac{\tilde\mu_2}{\tilde\mu_1}$. This is in contradiction with allocation inequalities since advertiser 2 wins $q_2$. Therefore, we proved with $T_1=0.4$ and $T_2=0.7$; the equilibrium is unique such that advertiser 1 wins $q_1$ and $q_2$. 

To finish the proof, we show that if advertiser 1 increases their tCPA constraint to $T'_1=0.6$, an equilibrium exists such that advertiser 1 only wins $q_1$.
Let $\mu_1'=1$ and $\mu_2'=2.38$. Then 
\begin{enumerate}
    \item Allocation: advertiser 1 wins $q_1$ 
\[\frac{v_1(q_1)}{v_2(q_1)}\geq \frac{\mu'_2}{\mu'_1}=2.38\geq\frac{v_1(q_2)}{v_2(q_2)}\geq  \frac{v_1(q_3)}{v_2(q_3)}.\]
    \item tCPA constraints are satisfied:
    \[79.1=T_2(v_2(q_3)+v_2(q_2))\geq \mu'_1 (v_1(q_3)+v_1(q_2))=50,\qquad\text{and }\]
    \[28=T_1'v_1(q_1)\geq \mu_2'v_2(q_1)=23.8.\]
\item Best response: non of the advertiser can win more queries if they increase their multiplier:
    \[86.1=T_2(v_2(q_1)+v_2(q_2)+v_2(q_3))< \mu_1'(v_1(q_1)+v_1(q_2)+v_1(q_3))=90,\qquad\text{and}\]
    \[42=T'_1(v_1(q_1)+v_1(q_2))< \mu_2' (v_2(q_1)+v_2(q_2))=54.74,\] 
    \[54=T'_1(v_1(q_1)+v_1(q_2)+v_1(q_3))< \mu_2' (v_2(q_1)+v_2(q_2)+v_2(q_3))=292.74.\] 
\end{enumerate} 
Note that in none of the equilibrium above, uniform bidding led to a tie. So, by Proposition~\ref{prop: uniform bidding}, which we prove next, uniform bidding is in fact optimal. 
\end{proof}

To finish the proof of Theorem~\ref{thm: second price Warm up}, it remains to prove that uniform bidding is optimal for SPA. Using Theorem 1 \cite{Aggarwal19} we get that uniform bidding is almost optimal up to a loss of 2 queries. However, their proof implies optimality in the examples we considered   Theorem~\ref{thm: second price Warm up}.

\begin{proposition}\label{prop: uniform bidding}
Consider SPA with queries as in Table~\ref{tab:budget-second price} or Table~\ref{tab:tcpa-second price}, then
    uniform bidding is an optimal best-response strategy.
\end{proposition}
\begin{proof}{Proof.}
We begin by introducing an artificial tie-breaking rule in the auction. We assume that in the event of a tie, the auction assigns the query to each of the advertisers if it is optimal for them. It is conceivable that a tie could occur where both advertisers win the same query in their optimal solutions. In this case, we assume that the auction can assign the query to both advertisers. With this tie-breaking rule in place, we proceed to show that uniform bidding is optimal.

    For this purpose, we borrow the notations of \cite{Aggarwal19}. 
They prove uniform bidding is almost optimal by taking a primal-dual approach. In our proof, we revisit their arguments to establish that uniform bidding is indeed optimal for single-slot queries when there are no ties.

 The optimization an auto bidder faces in single slot discrete queries is as in \eqref{eq: auto-bidder agent-descrete}:
 \begin{align}\label{eq: auto-bidder agent-descrete}
    \max& \sum_{q\in Q}x_{a}(q)v_{a}(q)\\
    \text{s.t.}& \sum_{q\in Q}x_{a}(q)p_{a}(q)\leq B_a+T_a\sum_{q\in Q}x_{a}(q)v_{a}(q)\label{eq: constraint}\\
 &  x_{a}(q)\in \{0,1\} \qquad    \forall q\in Q,
\end{align}
where $p_a$ is the payment function that depends on the other advertiser's bid. One can relax this integer program to a linear program and consider its dual, 
 \begin{align}\label{eq: dual auto-bidder agent-descrete}
    \min& \sum_{q\in Q}\delta_a(q)+\alpha_a B_a\\
    \text{s.t.}& \delta_a(q)+\alpha_a (p_a(q)-T_av_a(q))\geq v_a(q)\forall q\in Q\label{eq: dual constraint}\\
 &  \delta_{a}(q)\geq 0 \qquad    \forall q\in Q,\\
 &  \alpha_a\geq 0.
\end{align} 
Following the proof of  Theorem~1 in \cite{Aggarwal19}, with single-slot truthful auction, the uniform bidding $bid_a(q)=\frac{v_a(q)+\alpha_a T_av_a(q)}{\alpha_a}$ results in optimal queries except for the case that $\delta_a(q)=0$.  In this case, Lemma 1 of \cite{Aggarwal19} states that if $\delta_a(q)=0$ and $bid_a(q)<p_a(q)$ then $x_a(q)=0$, and query $q$ is not allocated to $a$ in any situation. But this lemma doesn't state anything about the case that $\delta_a(q)=0$ and $bid_a(q)=p_a(q)$, which means there was a tie in the advertisers' bids when the auction is SPA. Further,  by Lemma 1, part 5 in \cite{Aggarwal19}, a tie can happen for at most one query. 
By our assumption about tie-breaking, if a tie happens, the auction assigns the query to each of the advertisers with respect to the optimal $x_a(q)$ solution, and hence uniform bidding with this tie-breaking rule leads to an optimal solution. 
 
To conclude the proof, we observe that in the equilibrium considered in both examples of tCPA and budget-constrained advertisers in Theorem~\ref{thm: second price Warm up}, there were no ties. Hence, uniform bidding is optimal.

\end{proof}

\section{First-price Pacing Equilibrium}\label{apendix: pacing}
\begin{proof}{Proof of Lemma~\ref{lm: conitzer}.}
We follow the same steps of the proof as in \citet{conitzer2022pacing} for tCPA-constrained advertisers.
    Consider two sets of feasible bidding multipliers $\mathbf \mu$ and $\mathbf \mu'$. We will show that $\mathbf \mu^*=\max(\mathbf \mu,\mathbf\mu')$ is also feasible, where $\max$ is the component-wise maximum of the bidding profiles for $n$ advertisers. 
    
    Each query $q$ is allocated to the bidder with the highest pacing bid. We need to check that constraint \eqref{eq: constraint} is satisfied. Fix advertiser $a$. Its multiplier in $\mathbf \mu^*$ must also be maximum in one of $\mathbf \mu$, or $\mathbf\mu'$. without loss assume $\mu^*_a=\mu_a$. Then the set of queries that $a$ wins with bidding profile $\mathbf \mu^*$ ($X^*_a$) must be a subset of queries it wins n $\mathbf \mu$ ($X_a$) since all other advertisers' bids have either remained the same or increased. On the other hand, the cost of queries $a$ wins stays the same since it's a FPA. Since constraint  \eqref{eq: constraint} is feasible for bidding multipliers $\mathbf \mu$ we must have
    \[(\mu_a-T_a)\int_{q\in X} v_a(q)\leq B_a.\] But then since $X^*\subseteq X$, we have  as well
     \[(\mu_a-T_a)\int_{q\in X^*} v_a(q)=(\mu_a^*-T_a)\int_{q\in X^*} v_a(q)\leq B_a,\]
     which implies $\mathbf\mu^*$ is a feasible strategy.

To complete the proof we need to show the strategy that all advertisers take the maximum feasible pace $\mu^*_a=\sup\{\mu_a|\mu \text{ is feasible}\}$ results in equilibrium. To see this, note that if an advertiser's strategy is not best-response, they have the incentive to increase their pace with its constraints remaining satisfied. But then this would result in another feasible pacing strategy and is in contradiction with the choice of the highest pace $\mu_a^*$. A similar argument also shows that equilibrium is unique. Assume there exists another pacing equilibrium where an advertiser $a$ exists such that its pace is less than $\mu^*_a$. Then by increasing their pace to $\mu_a^*$ they will get at least as many queries as before, so $\mu_a^*$ is the best-response strategy.
\end{proof}
\section{Proofs for Truthful Auctions}\label{appendix: proofs truthful}
We start by the following observation, which follows by applying Assumption~\ref{assume: fair and truthful} to reformulate the allocation function in the case of two advertisers as a function of a single variable.
\begin{claim}
The probability of allocating each query is a function of the ratio of bids, i.e.,
there exists a non-decreasing function $g:\mathbb R^+\to[0,1]$ such that the followings hold.\footnote{Notice that the function $g$ is measurable since is non-decreasing.}
\begin{enumerate}
    \item $x_1(b_1(q),b_2(q))=g(\frac{b_1(q)}{b_2(q)})$,
    \item $g(z)+g(1/z)=1$,
    \item $g(0)=0$.
\end{enumerate}
\end{claim}
For example, SPA satisfies the above claim with $g(z)=1$ when $z=\frac{b_1(q)}{b_2(q)}\geq 1$. 
We are ready to prove  Theorem~\ref{thm: equilibrium myerson}, which follows the similar steps of Lemma~\ref{lm: equilibrium first price}.
\begin{proof}{Proof of Theorem~\ref{thm: equilibrium myerson}.}
By Claim~\ref{claim: uniform bididng}, there exists $\mu_1$ and $\mu_2$ such that advertiser $a$ bids $z_av_a(q)$ on each query.
Therefore, we can write the budget constraint for bidder 1 as,
\[B_1=\int_0^1p_1(b_1(q),b_2(q))dq=\int_0^1\mu_1v_1(q)g\Big(\frac{v_1(q)}{v_2(q)}\frac{\mu_1}{\mu_2}\Big)dq-\int_0^1\int_0^{\mu_1v_1(q)}g\Big(\frac{x}{v_2(q)\mu_2}\Big)dxdq\]
Next,  with a change of variable $x=v_1(q)y$ we have
\[B_1=\int_0^1\mu_1v_1(q)g\Big(\frac{v_1(q)}{v_2(q)}\frac{\mu_1}{\mu_2}\Big)dq-\int_0^1\int_0^{\mu_1}g\Big(\frac{v_1(q)}{v_2(q)\mu_2}y\Big){v_1(q)}dydq.\]
As before, let $h(q)=\frac{v_1(q)}{v_2(q)}$. Then let $z=h(q)$, we have $dq=dh^{-1}(z)=\frac{1}{h'(h^{-1}(z))}dz$. So, 
\[B_1=\int_0^\infty \mu_1v_1(h^{-1}(z))g\Big(\frac{z\mu_1}{\mu_2}\Big)\frac{1}{h'(h^{-1}(z))}dz-\int_0^\infty\int_0^{\mu_1}g\Big(\frac{z}{\mu_2}y\Big){v_1(h^{-1}(z))}dy\frac{1}{h'(h^{-1}(z))}dz.\]
Define $f(z)=\frac{v_2(h^{-1}(z))}{h'(h^{-1}(z))}=\frac{1}{z}\frac{v_1(h^{-1}(z))}{h'(h^{-1}(z))}$. Then we have
\[B_1=\int_0^\infty \mu_1zf(z)g\Big(\frac{z\mu_1}{\mu_2}\Big)dz-\int_0^\infty\Big(\int_0^{\mu_1}g\Big(\frac{z}{\mu_2}y\Big)dy\Big)zf(z)dz.\]
Similarly,
\[B_2=\int_0^\infty \mu_2v_2(h^{-1}(z))(1-g\Big(\frac{z\mu_1}{\mu_2}\Big))\frac{1}{h'(h^{-1}(z))}dz-\int_0^\infty\int_0^{\mu_2}g\Big(\frac{y}{\mu_1z}\Big){v_2(h^{-1}(z))}dy\frac{1}{h'(h^{-1}(z))}dz.\]
\[B_2=\int_0^\infty \mu_2f(z)(1-g\Big(\frac{z\mu_1}{\mu_2}\Big))dz-\int_0^\infty\int_0^{\mu_2}g\Big(\frac{y}{\mu_1z}\Big)dyf(z)dz.\]

Next, we find the implicit function to derive $r=\frac{\mu_2}{\mu_1}$.
By change of variable, we have the following two equations:
\[\frac{B_1}{\mu_1}=\int_0^\infty zf(z)g(z/r)dz-r\int_0^\infty \Big(\int_0^{z/r}g(w)dw\Big)f(z)dz.\]
\[\frac{B_2}{\mu_2}=\int_0^\infty f(z)(1-g(z/r))dz-\frac{1}{r}\int_0^\infty\Big( \int_0^{r/z}g(w)dw \Big)zf(z)dz.\]
The implicit function for $r$ is the following:
\[\frac{B_1}{B_2}=\frac{\int_0^\infty f(z)\Big( zg(z/r)- r\int_0^{z/r}g(w)dw\Big)dz}{\int_0^\infty f(z)\Big(r(1-g(z/r)-z\int_0^{r/z}g(w)dw \Big)dz}.\]
Recall the payment rule in Assumption~\ref{asum :truthful}, this can be re-written as
\[\frac{B_1}{B_2}=\frac{\int_0^\infty rf(z)p_1(z/r,1)dz}{\int_0^\infty zf(z)zp_1(r/z,1)dz},\]
which finishes the proof for the budget-constrained advertisers.

Now, consider two tCPA-constrained advertisers. Following the same argument as above, we get the following from the tightness of tCPA constraints

\[T_1\int_0^\infty zf(z)g\Big(\frac{z\mu_1}{\mu_2}\Big)dz=\int_0^\infty \mu_1zf(z)g\Big(\frac{z\mu_1}{\mu_2}\Big)dz-\int_0^\infty\Big(\int_0^{\mu_1}g\Big(\frac{z}{\mu_2}y\Big)dy\Big)zf(z)dz,\]
and,
\[T_2\int_0^\infty f(z)(1-g\Big(\frac{z\mu_1}{\mu_2}\Big))dz=\int_0^\infty \mu_2f(z)(1-g\Big(\frac{z\mu_1}{\mu_2}\Big))dz-\int_0^\infty\int_0^{\mu_2}g\Big(\frac{y}{\mu_1z}\Big)dyf(z)dz.\]
By dividing both sides of the equations we get the desired results.
\end{proof}
Now, to prove the main theorem, we need to show that the values accrued by advertisers is monotone in $\mu_1/\mu_2$.
\begin{claim}
Let $\mu_i$ be the optimal bidding multiplier for advertiser $i$. Given the assumptions in Theorem \ref{thm: truthful}, the value obtained by advertiser $1$ is increasing in $r=\frac{\mu_1}{\mu_2}$.
\end{claim}
\begin{proof}{Proof.}
Following the proof of Theorem~\ref{thm: equilibrium myerson} we can write the value obtained by advertiser $i$ as 
\[V_1(B_1,B_2)=\int_0^\infty f(z)zg(rz)dz,\]
where $r$ is the answer to the implicit function stated in  Theorem~\ref{thm: equilibrium myerson}. Monotonicity of $V_1(B_1,B_2)$ as a function of $r$ follows from the fact that $g$ is a monotone function.
\end{proof}

\end{appendix}

\end{document}